\documentclass{llncs} 

\usepackage{amssymb}
\usepackage{textcomp}
\usepackage{amsmath}
\usepackage{algorithm}
\usepackage{algpseudocode}
\usepackage{comment}
\usepackage{graphicx}
\usepackage{color}
\newcommand{\Endproof}{\hfill$\Box$\\}

\newcommand{\ket}[1]{|#1\rangle}

\begin{document}

\title{Quantum Property Testing Algorithm for the Concatenation of Two Palindromes Language}
\author{Kamil Khadiev\inst{1} \and
Danil Serov\inst{1}}
\authorrunning{K. Khadiev and D. Serov}
%
\institute{Institute of Computational Mathematics and Information Technologies, Kazan Federal University, Kazan, Russia\\
 \email{kamilhadi@gmail.com} }

\maketitle

\begin{abstract}
In this paper, we present a quantum property testing algorithm for recognizing a context-free language that is a concatenation of two palindromes $L_{REV}$. The query complexity of our algorithm is $O(\frac{1}{\varepsilon}n^{1/3}\log n)$, where $n$ is the length of an input. It is better than the classical complexity that is $\Theta^*(\sqrt{n})$. 

At the same time, in the general setting, the picture is different a little. Classical query complexity is $\Theta(n)$, and quantum query complexity is $\Theta^*(\sqrt{n})$. So, we obtain polynomial speed-up for both cases (general and property testing). 

\end{abstract}

\section{Introduction}
In this paper, we investigate property testing \cite{rs96,ggr98} that is a relaxation of the standard notion of a decision problem. A property testing algorithm distinguishes between inputs with a certain property and inputs that are far from any input that has the property. By ``far'' we mean a large Hemming distance. More specifically, for a given property $\alpha$, a testing algorithm should accept an input that has the property, and
reject if the input is $\varepsilon$-far from any input with the property. Here $\varepsilon$-far means that the Hemming distance is at least $\varepsilon\cdot n$, where $n$ is the length of the input. In \cite{r2001,f2004}, one can find examples of testing algorithms whose query complexity is sublinear or independent of the
input size.

Researchers investigate formal languages with respect to the property testing. Alon, Krivelevich, Newman, and Szegedy \cite{akns2001} presented a property testing algorithm for any regular language $L$ with query complexity $O^*(1/\varepsilon)$ that does not depend on the input size. Here $O^*$ hides constant and log factors. Newman \cite{n2000} extended this result to properties defined
by bounded-width branching programs. 

At the same time, Alon et al. \cite{akns2001} showed that the situation for context-free languages is completely different. Context-free languages are not testable even in time square root in the input size. As an example, the context-free language $L_{REV}=\{uu^rvv^r: u,v\in \Sigma^*\}$ that is a concatenation of two palindromes, where $\Sigma$ is a finite size alphabet (binary as an example), and $u^r$ is  a reverse of $u$. For the language, they proved $\Omega(\sqrt{n})$ lower bound for query complexity, where $n$ is a length of an input. Parnas, Ron, and Rubinfeld presented a property testing algorithm that almost reaches this lower bound. Its  query complexity is $O\left(\frac{1}{\varepsilon}\sqrt{n}\log n\right)$.

Buhrman, Fortnow, Newman, and R{\"o}hrig \cite{bfnr2008} introduced quantum property testing. They developed quantum property testing algorithms for some problems that are  better than classical counterparts in terms of query complexity. A nice survey on quantum property testing can be found in \cite{md2016}. At the same time, context-free languages like $L_{REV}$ were not considered. We are interested in developing a quantum property testing algorithm for a context-free language that is better than the classical lower bound.

There are many examples of quantum algorithms \cite{nc2010,a2017,k2022lecturenotes,aazksw2019part1} that are faster than classical counterparts  \cite{dw2001,quantumzoo} in the general setting (not property testing). 
Problems for strings are examples of such problems \cite{aaksv2022,aj2021,kkmsy2022,ki2019,kbcw2024,ke2022,kiv2022,kk2021,kb2022,kr2021b,kr2021a,kszm2022,l2020,l2020conf,m2017}.

A new interest in recognizing formal languages, including context-free languages, is started from the paper of Aaronson, Grier, and Schaeffer \cite{ags2019}. Dyck language was investigated by different researches \cite{abikkpssv2020,abikkpssv2023,bps2019,kk2021}. Other formal languages were explored in papers \cite{bfk2015,ckksw2022}. 

In this paper, we present a quantum property testing algorithm for recognizing $L_{REV}$ language that has $O(\frac{1}{\varepsilon}n^{1/3}\log n)$ query complexity. It shows quantum speed-up and it is better than the classical lower bound $\Omega(\sqrt{n})$. For this result, we use the meet-in-the-middle technique and Grover's search algorithm \cite{g96,bbht98}.

At the same time, in the general setting (not a property testing algorithm), we show that the problem has $\Theta(n)$ classical query complexity; and $\Theta^*(\sqrt{n})$ quantum query complexity. We present a quantum lower bound $\Omega(\sqrt{n})$, and a quantum algorithm with query complexity $O\left(\sqrt{n}(\log n)^2\right)$. So, we obtain almost quadratic speed-up.
We see that in the general setting, the classical lower bound differs from the property testing setting. At the same time, we see quantum speed-up for both cases.

The structure of this paper is the following.
Section \ref{sec:prelims} describes some conventional notions for quantum computation.
Section \ref{sec:algo-gen} provides quantum and classical algorithms and lower bounds for general setting. The quantum property testing algorithm is given in
subsection \ref{sec:algo}. The final Section \ref{sec:concl} concludes the paper and contains open questions. 

\section{Preliminaries}\label{sec:prelims}
For a string $u=(u_1,\dots,u_M)$, let $|u|=M$ be a length of the string, and let $u^r=(u_M,\dots,u_1)$ be the reverse of the string $u$.

Let us formally define the Two Palindromes Concatenation problem.

Suppose $\Sigma$ is a finite-size alphabet. Let $L_{REV}=\{uu^rvv^r: u,v\in\Sigma^*\}$ be a language of concatenations for two palindromes. We assume that in the definition of  $L_{REV}$  $u$ and $v$ are not empty strings.  For simplicity, in the paper, we assume that the alphabet is binary, $\Sigma=\{0,1\}$. At the same time, all results are correct for any finite-size alphabet.

For an integer $n>0$, let the function $\texttt{REV}_{n}:\Sigma^n\to\{0,1\}$ be such that $\texttt{REV}_{n}(x)=1$ iff $x\in L_{REV}$.  

For an integer $n>0$ and a non-negative $\varepsilon<1$, let $\texttt{REV}_{n,\varepsilon}$ be a property testing problem such that for an input $x=(x_0,\dots,x_{n-1})$ there is a promise that if $x$ is not in $L_REV$, then $x$ is at least $\varepsilon \cdot n$ far from the closest word from $L_{REV}$. Formally, if $x$ is not in $L_REV$, then for any $u\in L_{REV}$ we have $|\{i:x_i\neq u_i\}|\geq \varepsilon \cdot n$.


In the paper, we use a {\em trie} (prefix tree) data structure \cite{d59,b98,b2008,knuth73}.
It is a tree that allows us to add a string $s$  and check whether $s$ is in the tree with running time $O(|s|)$. The data structure implements a ``set of strings'' data structure. Let us have the following operations with a trie $T$:
\begin{itemize}
    \item $\textsc{InitTrie()}$ returns an empty trie. The running time of the operation is $O(1)$.
    \item $\textsc{AddToTrie(T,s)}$ adds a string $s$ to the trie $T$. The running time of the operation is $O(|s|)$.
    \item $\textsc{Contains(T,s)}$ returns $1$ if a string $s$ belongs to the trie $T$, and $False$ otherwise. The running time of the operation is $O(|s|)$.
\end{itemize}

\subsection{Quantum query model}
One of the most popular computation models for quantum algorithms is the query model.
We use the standard form of the quantum query model. 
Let $f:D\rightarrow \{0,1\},D\subseteq \{0,1\}^M$ be an $M$ variable function. Our goal is to compute it on an input $x\in D$. We are given oracle access to the input $x$, i.e. it is implemented by a specific unitary transformation usually defined as $\ket{i}\ket{z}\ket{w}\mapsto \ket{i}\ket{z+x_i\pmod{2}}\ket{w}$, where the $\ket{i}$ register indicates the index of the variable we are querying, $\ket{z}$ is the output register, and $\ket{w}$ is some auxiliary work-space. An algorithm in the query model consists of alternating applications of arbitrary unitaries which are independent of the input and the query unitary, and a measurement at the end. The smallest number of queries for an algorithm that outputs $f(x)$ with probability $\geq \frac{2}{3}$ on all $x$ is called the quantum query complexity of the function $f$ and is denoted by $Q(f)$.
We refer the readers to \cite{nc2010,a2017,aazksw2019part1,k2022lecturenotes} for more details on quantum computing. 

In this paper, we are interested in the query complexity of the quantum algorithms. We use modifications of Grover's search algorithm \cite{g96,bbht98} as quantum subroutines. For these subroutines, time complexity can be obtained from query complexity by multiplication to a log factor \cite{ad2017,g2002}.   
\section{The General Case}\label{sec:algo-gen}
In this section, we consider the $\texttt{REV}_{n}$ problem. Here we show quantum upper and lower bounds that are almost equal up to log factors.
\subsection{Quantum and Classical Algorithms}
Let us start with the upper bound.

Firstly, let us show one useful property. For the input string $x=(x_0,\dots,x_{n-1})$, let $\bar{x}=(x_1,\dots,x_{n-1})$ be the string $x$ without the first symbol. Let $\hat{x}=(x_0,\dots,x_{n-2})$ be the string $x$ without the last symbol.  Let 
\[y(x)=\bar{x}\circ \hat{x}=(x_1,\dots,x_{n-1},x_0,\dots,x_{n-2}),\] 
where $\circ$ is the concatenation operation. Then, we have the following result
\begin{lemma}\label{lm:substr}
A string $x\in L_{REV}$ if and only if $y(x)$ contains $x^r$ as a substring.
\end{lemma}
\begin{proof}
    Assume that $x\in L_{REV}$. It means that we can find two strings $u$ and $v$ such that $x=uu^rvv^r$. Note that $\hat{v^r} = \bar{v}^r$ which means we can either remove the last symbol of a reversed string, or remove the first symbol of the original string and then reverse. Hence, the string $y(x)$ has the following form
    \[y(x)=\bar{u}u^rvv^ruu^rv\bar{v}^r\]
    We can see that it has $x^r=vv^ruu^r$ as a substring because $y(x)=\bar{u}u^r
    \circ x^r \circ v\bar{v}^r$.

    Assume that $x^r$ is a substring of $y(x)$. Let $n=|x|$, and $k=|y(x)|=2n-2$. Assume that $x^r$ starts in $y(x)$ from a position $i\leq n-1$. It means
    that $y(x)_i=(x^r)_0$. At the same time, $(x^r)_0=x_{n-1}$ because $x^r$ is the reverse of $x$. We also can say, that $y(x)_i=x_{i+1}$.
    So, we have
    \[x_{i+1}=y(x)_{i}=(x^r)_0=x_{n-1},\] 
    \[x_{i+2}=y(x)_{i+1}=(x^r)_1=x_{n-2},\] 
    \[x_{i+3}=y(x)_{i+2}=(x^r)_2=x_{n-3},\]
    \[...\]
    \[x_{n-1}=y(x)_{n-2}=(x^r)_{n-i-2}=x_{i+1}\]
    Therefore, $(x_{i+1},\dots,x_{n-1})$ is a palindrome. Let it be $vv^r$ for some string $v$ (See Figure \ref{fig:palv}).

    \begin{figure}
    \centering
    \includegraphics[width=\textwidth]{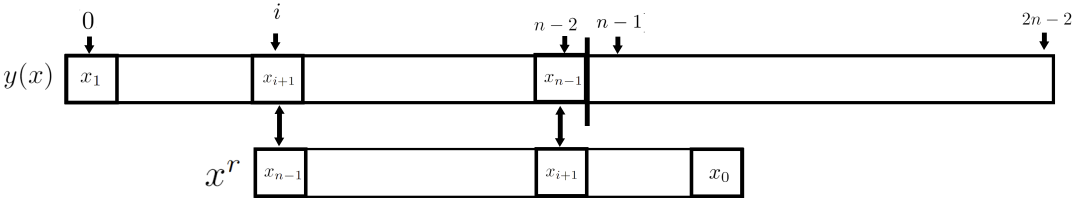}
    \caption{The string $x^r$ is a substring of $y(x)$ and it starts from position $i$. We can see that $(x_{i+1},\dots,x_{n-1})$ is a palindrome.}
    \label{fig:palv}
\end{figure}
    
     Since $x^r$ starts from the $i$-th symbol of $y(x)$, we have $y(x)_{n-1}=(x^r)_{n-i-1}$. At the same time, $y(x)_{n-1}=x_0$; and $(x^r)_{n-i-1}=x_{i}$.  So, we have
    \[x_0=y(x)_{n-1}=(x^r)_{n-i-1}=x_{i},\] 
    \[x_1=y(x)_{n}=(x^r)_{n-i}=x_{i-1},\]
    \[x_2=y(x)_{n+1}=(x^r)_{n-i+1}=x_{i-2},\]
    \[...\]
    \[x_{i}=y(x)_{n+i-1}=(x^r)_{n-1}=x_{0}\]
    Therefore, $(x_0,\dots,x_{i})$ is a palindrome. Let it be $uu^r$  for some string $u$ (See Figure \ref{fig:palu}).
    \begin{figure}
    \centering
    \includegraphics[width=\textwidth]{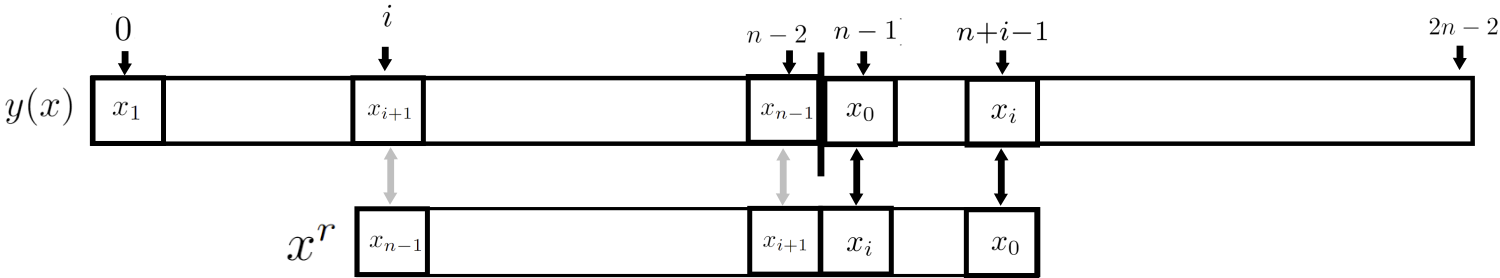}
    \caption{The string $x^r$ is a substring of $y(x)$ and it starts from position $i$. We can see that $(x_{0},\dots,x_{i})$ is a palindrome.}
    \label{fig:palu}
\end{figure}

    So, we can say that $x=(x_0,\dots,x_{i})\circ (x_{i+1},\dots,x_{n-1})=uu^rvv^r\in L_{REV}$.
\Endproof
\end{proof}

In fact, we do not construct $y(x)$. To  access the symbol $y(x)_i$,  we use a function $\textsc{Get}(i)$ that returns $x_{i+1}$ if $i<n-1$, and returns $x_{i-n+1}$ if $\geq n-1$. The complexity of this function is $O(1)$ if we use array-like data structures, but not Linked List data structure.

In the classical case, the substring problem can be solved using the Knuth–Morris–Pratt algorithm \cite{cormen2001,kmp77}. The complexity of the algorithm is $O(k+n)$, where $n=|x^r|$, and $k=|y(x)|=2n-2$. So, the complexity is $O(n)$.

In the quantum case, we can solve the problem using the Ramesh-Vinay algorithm \cite{rv2003}.  The complexity of the algorithm is $O\left(\sqrt{k}\log{\sqrt{\frac{k}{n}}}\log n + \sqrt{n}(\log n)^2\right)$. 

The final complexity of both algorithms is presented in the next theorem.
\begin{theorem}
    For a positive integer $n$ and the problem $\texttt{REV}_n$, there is a classical algorithm that works with query and time complexity $O(n)$; and a quantum algorithm that works with query complexity $O\left(\sqrt{n}(\log n)^2\right)$ and two-side error probability strictly less than $0.5$.
\end{theorem}
\begin{proof}
    Due to Lemma \ref{lm:substr}, the problem is equivalent to searching $x^r$ in $y(x)$. Let $k=|y(x)|=2n-2$, and $n=|x|=|x^r|$.
    
    In the classical case, time and query complexity of Knuth–Morris–Pratt algorithm \cite{cormen2001,kmp77} is $O(n+k)=O(n+2n-2)=O(n)$. 

    In the quantum case, the query complexity of the Ramesh-Vinay algorithm \cite{rv2003} is
    \[O\left(\sqrt{k}\log{\sqrt{\frac{k}{n}}}\log n + \sqrt{n}(\log n)^2\right)=\]\[
    O\left(\sqrt{2n-2}\log{\sqrt{\frac{2n-2}{n}}}\log n + \sqrt{n}(\log n)^2\right)=\]\[
    O\left(\sqrt{n}\log n + \sqrt{n}(\log n)^2\right)= O\left(\sqrt{n}(\log n)^2\right).\]
    The error probability is strictly less than $0.5$ due to \cite{rv2003}.
\Endproof
\end{proof}
\subsection{Lower Bound}
Let us present the lower bound for the $\texttt{REV}_n$ problem. In the next theorem, we show that the problem is at least as hard as an unstructured search among $n$ elements. 
\begin{theorem}
    The lower bound for quantum query complexity of $\texttt{REV}_n$ is $\Omega(\sqrt{n})$, and for classical (randomized or deterministic) query complexity is $\Omega(n)$. 
\end{theorem}
\begin{proof}
Assume that $n=2t+2$ for some integer $t$.
Let us consider only inputs of two forms:
\begin{itemize}
    \item $\sigma=(0,\dots,0)$ is a $0$-string. Let $u=(0)$, and $v=(0,\dots,0)$. Here $|u|=1$, and $|v|=t$. In that case, $u=u^r$, and $v=v^r$. So, we can say that $\sigma=uu^rvv^r$, and $\texttt{REV}_n(\sigma)=1$.
    \item For $i\in\{0,\dots,n-1\}$, the string $\gamma^i=(0,\dots,0,1,0,\dots 0)$ has $1$ on the position $i$, and $0$s on other positions. There is only one position with a $1$-value, and it has not a symmetric pair. Therefore, $\texttt{REV}_n(\gamma^{i})=0$.
\end{itemize}
Distinguishing between two cases $\sigma$ and $\gamma^i$ is equivalent to searching $1$ among $n$ elements. 

Assume, that we have a quantum algorithm with quantum query complexity $o(\sqrt{n})$ or a classical algorithm with query complexity $o(n)$. Then, we can distinguish between two cases $\sigma$ and $\gamma^i$, and find $1$ among $n$ elements with proposed complexity. This claim contradicts the lower bound for unstructured search \cite{bbbv1997} that is $\Omega(\sqrt{n})$ in the quantum case, and $\Omega(n)$ in the classical case.  
\Endproof
\end{proof}

Finally, we see that the classical complexity for the problem is $\Theta(n)$, and the quantum complexity is $\Theta^*(\sqrt{n})$, where $\Theta^*$ hides logarithmic factors. So, we obtain an almost quadratic speed-up for this problem. 

\section{The Property Testing Case}\label{sec:algo}
In this section, we consider the $\texttt{REV}_{n,\varepsilon}$ problem.
Here, we use ideas from \cite{pr2003} paper that provides a randomized algorithm for the problem.

The classical upper bound \cite{pr2003} for the problem is $O\left(\frac{1}{\varepsilon}\sqrt{n}\log n\right)$, and the lower bound \cite{akns2001} is $\Omega(\sqrt{n})$. We can see, that in the property testing case, we have significant improvement. A situation simulation happens for quantum algorithms. 

Firstly, let us discuss some observations on properties of a word $x=(x_0,\dots,x_{n-1})=uu^rvv^r$ from the language $L_{REV}$.
Let us consider two indexes $i,j\in\{0,\dots,n-1\}$, and assume that $i<j$ without limiting the generality of the foregoing.

We say that they are in symmetric positions with respect to $uu^r$ if $i,j<2|u|$, and there is an integer $\delta\leq |u|$ such that $i+\delta=|u|-1$, and $j-\delta=|u|$. In other words, they are symmetric with respect to the middle of the palindrome $uu^r$ (See Figure \ref{fig:symu}).
\begin{figure}
    \centering
    \includegraphics[width=\textwidth]{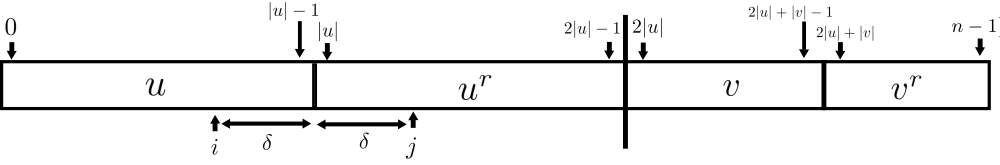}
    \caption{Indexes $i$ and $j$ are symmetric with respect to the middle of the palindrome $uu^r$.}
    \label{fig:symu}
\end{figure}

We say that indexes $i$ and $j$ are in symmetric positions with respect to $vv^r$ if $i,j\geq 2|u|$, and there is an integer $\delta'\leq |v|$ such that $i+\delta'= 2|u|+|v|-1$, and $j-\delta'=2|u|+|v|$. In other words, they are symmetric with respect to the middle of the palindrome $vv^r$ (See Figure \ref{fig:symv}).
\begin{figure}
    \centering
    \includegraphics[width=\textwidth]{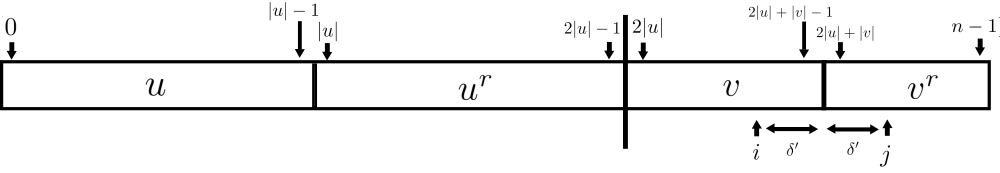}
    \caption{Indexes $i$ and $j$ are symmetric with respect to the middle of the palindrome $vv^r$.}
    \label{fig:symv}
\end{figure}

Let us consider any two indexes $i$ and $j$ that are in symmetric positions with respect to $uu^r$ or $vv^r$. In that case, we have the following lemma about these indexes:
\begin{lemma}\label{pr:symetric}
Indexes $i$ and $j$ are in symmetric positions with respect to $uu^r$ or $vv^r$ if and only if $i+j=2|u|-1$ $($mod $ n)$. Remember that $n=2|u|+2|v|$.
\end{lemma}
\begin{proof}
Without limiting the generality of the foregoing, we can assume that $i<j$.
Note that $0\leq i+j\leq 2n$ because $i,j\leq n$. Therefore,  $i+j=2|u|-1$ $($mod $ n)$ means  either $i+j=2|u|-1$ or $i+j=2|u|-1+n$. We have three cases. Let us consider each of them.

\begin{itemize}
    \item[Case 1.] 
Assume that $i,j<2|u|$.
Let $i$ and $j$ be in symmetric positions with respect to $uu^r$. So, there is $\delta$ such that $i+\delta = |u|-1$ and $j-\delta=|u|$. Hence, $i+j=|u|-1-\delta + |u|+\delta=2|u|-1=2|u|-1$ $($mod $ n)$.

Let $i+j=2|u|-1$ $($mod $ n)$. Let $\delta=|u|-1-i$. Then \[i=|u|-\delta-1, \mbox{ and }j=2|u|-1-i=2|u|-1-|u|+1+\delta=|u|+\delta.\]
Due to the definition, $i$ and $j$ are in symmetric positions with respect to $uu^r$.

\item[Case 2.]
Assume that $i, j \geq 2|u|$. Let $i$ and $j$ be in symmetric positions with respect to $vv^r$. So, there is $\delta'$ such that $i+\delta' = 2|u|+|v|-1$ and $2|u|+|v|=j-\delta'$. Hence,
\[i+j=2|u|+|v|-1-\delta' + 2|u|+|v|+\delta'=2|u|-1 + (2|u|+2|v|)\mbox{ $($mod } n).\]
Note that $2|u|+2|v|=n$. Therefore, $i+j=2|u|-1$ $($mod $ n)$.

Let $i+j=2|u|-1$ $($mod $ n)$. Let $\delta'=2|u|+|v|-1-i$. We remember that $n=2|u|+2|v|$. Then
\[i=2|u|+|v|-1-\delta',\mbox{ and }j=2|u|-1-i \mbox{ $($mod }n)\]
\[j=n+2|u|-1-i \mbox{ $($mod }n)\]
\[j=n+2|u|-1-2|u|-|v|+1+\delta'=n-|v|+\delta'=2|u|+2|v|-|v|+ \delta' \mbox{ $($mod }n)\]\[j=2|u|+|v|+ \delta' \mbox{ $($mod }n).\]
Note that $2|u|+|v|+\delta'=2|u|+|v|+2|u|+|v|-1-i\leq 2|u|+|v|+2|u|+|v|-1-2|u|=2|u|+2|v|-1<n$. Therefore, $j< n$, and we can say that $j=2|u|+|v|+ \delta'$.
Due to the definition, $i$ and $j$ are in symmetric positions with respect to $vv^r$.

\item[Case 3.]
Assume that $i<2|u|$, and $j\geq 2|u|$. In that case, $i$ and $j$ are not in symmetric positions with respect to $uu^r$ nor $vv^r$. Let us show that $i+j\neq 2|u|-1 ($ mod $n$). 

Let $i<n-j-1$, then $i+j<n-1$.
It means that $(i+j)$ mod $n=i+j$. 
At the same time, $i+j\geq j\geq 2|u|>2|u|+1$. Therefore $i+j\neq 2|u|-1$.

Let $i\geq n-j-1$. It means that $(i+j)$ mod $n=i-(n-j)<i-1<2|u|-1$.
Therefore $(i+j)$ mod $n\neq 2|u|-1$.
\end{itemize}
\Endproof
\end{proof}

So, we can say that for any $i$ and $j$ such that $i+j=2|u|-1$ $($ mod $n)$, we have $x_i=x_j$. 

For an integer $p\in\{0,\dots,n-1\}$, let us look at two indexes $(i-p)\mbox{ mod }n$ and $(j+p)\mbox{ mod }n$ where $i+j=2|u|-1$ $($ mod $n)$. We have
\[(i-p)+(j+p)=i+j=2|u|-1 (\mbox{ mod }n)\]
Therefore, $x_{(i-p)\mbox{ mod }n}=x_{(j+p)\mbox{ mod }n}$.

Let us consider the number $2|u|-1$. Due to the integer division rule and the statement $2|u|-1<2|u|+2|v|=n $, we can say that 
\[2|u|-1 = \alpha \cdot \lfloor n^{1/3}\rfloor + \beta,\] 
where $0\leq \alpha \leq \left\lfloor \frac{2|u|-1}{\lfloor n^{1/3}\rfloor}\right\rfloor \leq \left\lfloor \frac{n}{\lfloor n^{1/3}\rfloor}\right\rfloor\approx n^{2/3}$ and $0\leq \beta \leq \lfloor n^{1/3}\rfloor-1$. 

Let us consider two sets of integers that are 
\[{\cal I}_n=\{0,\dots, \lfloor n^{1/3}\rfloor-1\}\]
and
\[{ \cal J}_n=\{\alpha \cdot \lfloor n^{1/3}\rfloor,\mbox{ for }0\leq \alpha \leq \left\lfloor \frac{2|u|-1}{\lfloor n^{1/3}\rfloor}\right\rfloor\}=\]\[
\{0,\lfloor n^{1/3}\rfloor,2\lfloor n^{1/3}\rfloor,3\lfloor n^{1/3}\rfloor, \dots, \left\lfloor \frac{2|u|-1}{\lfloor n^{1/3}\rfloor}\right\rfloor\cdot \lfloor n^{1/3}\rfloor\}.\]

Note that $|{ \cal J}_n|=\left\lfloor \frac{2|u|-1}{\lfloor n^{1/3}\rfloor}\right\rfloor \leq \left\lfloor \frac{n}{\lfloor n^{1/3}\rfloor}\right\rfloor\approx n^{2/3}$, and $|{\cal I}_n|=\lfloor n^{1/3}\rfloor$.

We are ready to present one more lemma about these indexes.
\begin{lemma}\label{pr:int-div}
If $x=(x_0,\dots,x_{n-1})\in L_{REV}$, then there is $i\in{\cal I}_n$ and $j\in{\cal J}_n$ such that $x_{(i-p)\mbox{ mod }n}=x_{(j+p)\mbox{ mod }n}$ for any $p\in\{0,\dots,n-1\}$. 
\end{lemma}
\begin{proof}
If $x=(x_0,\dots,x_{n-1})\in L_{REV}$, then there is $u$ and $v$ from $\Sigma^{*}$ such that $x=uu^rvv^r$.

As we discussed before, $2|u|-1=\alpha \cdot \lfloor n^{1/3}\rfloor + \beta$, where $0\leq \alpha \leq \left\lfloor \frac{2|u|-1}{\lfloor n^{1/3}\rfloor}\right\rfloor$ and $0\leq \beta \leq \lfloor n^{1/3}\rfloor -1$. Let $i=\beta$ and $j=\alpha \cdot \lfloor n^{1/3}\rfloor$. 

Therefore, $i+j=2|u|-1$ and they are in symmetric positions with respect to $uu^r$ or $vv^r$ due to Lemma \ref{pr:symetric}. Hence, for any $p\in\{0,\dots,n-1\}$ we have $x_{(i-p)\mbox{ mod }n}=x_{(j+p)\mbox{ mod }n}$.
\Endproof
\end{proof}

We are ready to formulate the algorithm.

\subsection{Quantum Algorithm}
Let us present an algorithm for computing $\texttt{REV}_{n,\varepsilon}(x)$. The algorithm is based on the meet-in-the-middle technique \cite{k2022lecturenotes}(Section 8) that is widely used in algorithms design and cryptography \cite{dh77}. The main idea is to split a large set into two small parts, small enough for handling them. Similar ideas were used, for example, in \cite{bht97collarx,as2003}.

Let us consider the sets ${\cal I}_n$ and ${\cal J}_n$, and an integer $m=\frac{2}{\varepsilon}\log_2 n$.
\begin{itemize}
    \item[] {\bf Step 1.} We choose randomly $m$ numbers $p_1,\dots,p_m\in_R\{0,\dots,n-1\}$
    \item[] {\bf Step 2.} We add all strings $x^i=(x_{(i-p_1)\mbox{ mod }n},\dots,x_{(i-p_m)\mbox{ mod }n})$ for $i\in {\cal I}_n$ to a {\em trie} ( prefix tree) $T$.
    \item[] {\bf Step 3.} We search $j\in {\cal J}_n$ such that $\tilde{x}^j=(x_{(j+p_1)\mbox{ mod }n},\dots,x_{(j+p_m)\mbox{ mod }n})$ is presented in the {\em trie} $T$. We search them using Grover's search algorithm \cite{g96,bbht98}. We define a search function $f:{\cal J}_n\to \{0,1\}$ such that $f(j)=1$ iff $\tilde{x}^j$ is presented in $T$. The algorithm searches any $j_0$ such that $f(j_0)=1$.
\end{itemize}

If we found an index $j_0$ on Step 3, then there is $\tilde{x}^j=x^i$ and $i+j=2|u|-1$. Therefore, $\texttt{REV}_{n,\varepsilon}(x)=1$.

Assume that we have a $\textsc{GroverSearch}({\cal D}, f)$ procedure, that implements Grover's search algorithm for search space ${\cal D}$ and a function $f:{\cal D}\to \{0,1\}$. The algorithm finds $j_0\in{\cal D}$ such that $f(j_0)=1$. The algorithm works with $O(\sqrt{|{\cal D}|}\cdot T(f))$ query complexity, where $T(f)$ is the complexity of computing the function $f$. The error probability is at most $0.1$.
Assume that the procedure returns $True$ in the case of finding the element $j_0$ and $False$ otherwise. In our algorithm we use $\textsc{Contains}(T,\tilde{x}^j)$ function as $f(j)$. This function checks whether $\tilde{x}^j$ belongs to the trie $T$. The query complexity of the function is $O(m)$ due to properties of the trie data structure that were discussed in Section \ref{sec:prelims}. The main operations with the trie data structure are listed in Section \ref{sec:prelims}. 

The implementation of the algorithm is presented in Algorithm \ref{alg:main}; the complexity is analyzed in Theorem \ref{th:compl}.

\begin{algorithm}
\caption{The Quantum Algorithm for $\texttt{REV}_{n,\varepsilon}$ problem and input $x=(x_0,\dots,x_{n-1})$.}\label{alg:main}
\begin{algorithmic}
\State $m \gets \frac{2}{\varepsilon}\log_2 n$
\For{$i \in \{1,\dots,m\}$}
\State $p_i\gets_R\{0,\dots,n-1\}$\Comment{$p_i$ is chosen randomly}
\EndFor
\State $T\gets\textsc{InitTrie}()$ 
\For{$i \in \{1,\dots,\lfloor n^{1/3}\rfloor-1\}$}
\State  $x^i\gets (x_{(i-p_1)\mbox{ mod }n},\dots,x_{(i-p_m)\mbox{ mod }n})$
\State $\textsc{AddToTrie}(T,x^j)$
\EndFor

\State $Result\gets \textsc{GroverSearch}({\cal J}_n, \textsc{Contains}(T,\tilde{x}^j))$
\State \Return $Result$
\end{algorithmic}
\end{algorithm}

\begin{theorem}\label{th:compl}
The provided algorithm computes $\texttt{REV}_{n,\varepsilon}$ with $O\left(\frac{1}{\varepsilon}n^{1/3}\log n\right)$ query complexity and at most $\frac{1}{4}$ error probability.  \end{theorem}
\begin{proof}
Let us discuss complexity of the algorithm. Step 2 which is adding all strings $x^i$ requires $O(|{\cal I}_n|\cdot m)=O(n^{1/3}\cdot m)$ query complexity. Grover search works with $O(|{\cal J}_n|\cdot m)=O(\sqrt{n^{2/3}}\cdot m)=O(n^{1/3}\cdot m)$ query complexity. Note that $m=\frac{2}{\varepsilon}\log n$, that is why the total complexity is 
\[O\left(\frac{1}{\varepsilon}n^{1/3}\log n + \frac{1}{\varepsilon}n^{1/3}\log n\right)=O\left(\frac{1}{\varepsilon}n^{1/3}\log n\right)\]

As we discuss in Lemma \ref{pr:int-div}, if $x\in L_{REV}$, then there are $i\in {\cal I}_n$ and $j\in {\cal J}_n$ such that $x^i=\tilde{x}^j$ for any choice of $(p_1,\dots,p_m)$. Therefore, the algorithm finds the required $i$ and $j$ and returns the correct answer with an error probability at most $0.1$  because of the error probability for Grover's search algorithm.

Assume that $x$ is $\varepsilon\cdot n$ far from any word from the language $L_{REV}$.
In other words, $\texttt{REV}_{n,\varepsilon}(x)=0$. Let us show that with high probability we cannot find the $i$ and $j$ indexes.

For fixed $j$ and $i$, the probability of obtaining a position $k$ of equal symbols, $x^i_k=\tilde{x}^j_k$, is $(1-\varepsilon)$. The total error probability that is the probability of obtaining all $m$ equal positions is 
\[(1-\varepsilon)^m=(1-\varepsilon)^{\frac{2}{\varepsilon}\log_2 n}=\left((1-\varepsilon)^{\frac{1}{\varepsilon}}\right)^{2\log_2 n}<\left(\frac{1}{2}\right)^{2\log_2 n}=\left(2^{\log_2 n}\right)^{-2}=n^{-2}.\] 
Here $(1-\varepsilon)^{\frac{1}{\varepsilon}}<\frac{1}{2}$ for any $\varepsilon$ such that $0<\varepsilon< 1$.

For all pairs of $i$ and $j$, the probability of success on all pairs of strings is $(1-n^{-2})^{n^{1/3}\cdot n^{2/3}}>\frac{3}{4}$ for  $n\geq 4$  because $\lim_{n\to\infty}(1-n^{-2})^n=1$. So, with probability $\frac{3}{4}$ all elements in the search space of Grover's Search are $0$. If the whole input contains only zeroes, then Grover's search algorithm returns $0$ with probability $1$. So, the Total success probability in that case is at least $\frac{3}{4}$. The error probability is at most $\frac{1}{4}$.
\Endproof
\end{proof}
\section{Conclusion}\label{sec:concl}
In this paper, we present a quantum property testing algorithm for recognizing the context-free  language $L_{REV}$ that has $O(\frac{1}{\varepsilon}n^{1/3}\log n)$ query complexity. It is better than classical counterparts that have $\Theta^*(\sqrt{n})$ query complexity. At the same time, we do not know a quantum lower bound in the property testing setting. We have a feeling that it is $\Omega(n^{1/3})$.

At the same time, in the general setting, the picture is almost clear. Classical query complexity is $\Theta(n)$, and quantum query complexity is $\Theta^*(\sqrt{n})$. So, we almost obtain quadratic speed-up. The open question is to develop a quantum algorithm that reaches the lower bound without additional log factors.

The third open question is to investigate other context-free languages for quantum property testing.

{\bf Acknowledgements.} We thank Frédéric Magniez and Yixin Shen for useful discussions.

This paper has been supported by the Kazan Federal University Strategic Academic Leadership Program ("PRIORITY-2030").

\bibliographystyle{splncs03}
\bibliography{tcs}

\end{document}